\documentclass[twoside,leqno,twocolumn]{article}

\usepackage{ltexpprt}
\usepackage{graphicx}
\usepackage[latin9]{inputenc}
\usepackage{color}
\usepackage{verbatim}
\usepackage{textcomp}
\usepackage{amstext}
\usepackage{amssymb}
\usepackage{graphicx}
\usepackage{esint}

\makeatletter
\@ifundefined{date}{}{\date{}}
\makeatother

\begin{document}

\title{{\Large{Laplacian Spectral Properties of Graphs from Random Local
Samples}}}

\author{Zhengwei Wu\\
 \and Victor M. Preciado}
\maketitle
\begin{abstract}
\baselineskip=9pt The Laplacian eigenvalues of a network play an
important role in the analysis of many structural and dynamical network
problems. In this paper, we study the relationship between the eigenvalue spectrum
of the normalized Laplacian matrix and the structure of `local' subgraphs of the network.
We call a subgraph \emph{local} when it is induced by the set of nodes obtained from a breath-first search (BFS) of
radius $r$ around a node.
In this paper, we propose techniques to
estimate spectral properties of the normalized Laplacian matrix from a random collection
of induced local subgraphs. In particular, we provide
an algorithm to estimate the spectral moments of the normalized Laplacian matrix (the
power-sums of its eigenvalues). Moreover, we propose a technique,
based on convex optimization, to compute upper and lower bounds
on the spectral radius of the normalized Laplacian matrix from local subgraphs.
We illustrate our results studying the normalized Laplacian spectrum of a large-scale
online social network.
\end{abstract}

\section{Introduction}

Understanding the relationship between the structure of a network
and its eigenvalues is of great relevance in the field of network
science (see \cite{chung1997spectral}, \cite{olfati2006flocking}
and references therein). The growing availability of massive
databases, computing facilities, and reliable data analysis tools
has provided a powerful framework to explore this relationship for
many real-world networks. On the other hand, in many cases of practical
interest, one cannot efficiently retrieve and/or store the exact full
topology of a large-scale network. Alternatively, it is usually easy
to retrieve local samples of the network structure. In this paper,
we focus our attention on local sample of the network structure given in the form of a subgraph induced by the set of nodes obtained from
a breath-first search (BFS) of small radius $r$ around a particular
node.

We study the relationship between the normalized Laplacian spectrum of a graph
and a random collection of local subgraphs. We show how local structural information contained in these
localized subgraphs can be efficiently aggregated to infer global
properties of the normalized Laplacian spectrum. Our analysis reveals that certain
spectral properties, such as the so-called spectral moments (the power-sums
of the eigenvalues), can be efficiently estimated from a random collection
of localized subgraphs. Furthermore, applying recent results connecting
the classical moment problem and convex optimization, we propose a
series of semidefinite programs (SDP) to compute upper and lower bounds
on the Laplacian spectral radius from a collection of local structural
samples.

\subsection{Previous Work}

Studying the relationship between the structure of a graph and its
eigenvalues is the central topic in the field of algebraic graph theory
\cite{biggs1993algebraic},\cite{chung1997spectral},\cite{cvetkovic2010introduction},\cite{mohar1991laplacian}.
In particular, the spectrum of the Laplacian matrix has a direct connection to the behavior of several
networked dynamical processes, such as random walks \cite{lovasz1993random},
consensus dynamics \cite{olfati2006flocking}, and a wide variety
of distributed algorithms \cite{lynch1996distributed}.

In many cases of practical interest it is unfeasible to exactly retrieve
the complete structure of a network of contacts, making it impossible to compute the graph spectrum directly. However, in most cases one can easily retrieve local subgraphs obtained via BFS of small radius. To estimate
spectral properties from localized structural samples, researchers
have proposed a variety of random network models in which they can
prescribe features retrieved from these samples, such as the degree
distribution \cite{chung2003spectra},\cite{newman2001random}, local correlations \cite{newman2002assortative},\cite{pastor2001dynamical},
or clustering \cite{newman2009random}.

Although random networks are the primary tool to study the impact
of local structural features on spectral network properties \cite{chung2003spectra}, this approach presents a major flaw: Random network models
implicitly induce structural properties that are not directly
controlled in the model construction, but can have a strong influence on the eigenvalue spectrum.
For example, it is possible to find two networks having the same degree
distribution, but with very different eigenvalue spectra.

\subsection{Our contribution}

In this paper, we develop a mathematical framework, based on algebraic
graph theory and convex optimization, to study how localized samples
of the network structure can be used to compute spectral properties
of the normalized Laplacian matrix of (possibly weighted) graphs. The following are our main contributions:
\begin{itemize}
\item We develop a \emph{sublinear} time algorithm to estimate the spectral moments
(power-sums of the eigenvalues) of the normalized Laplacian matrix of a graph
from a random set of local subgraphs samples. In our analysis,
we use Hoeffding inequality to provide quality guarantees of our estimators as a function of the number of samples.
\item We provide a convex program to compute, in linear time\footnote{Our algorithm runs in linear time assuming that the size of the local subgraphs are much smaller than $n$.},
upper and lower bounds on the Laplacian spectral radius from a random
set of local subgraph samples. Our results are based on recent
results connecting the classical moment problem with semidefinite
programs (SDP).
\end{itemize}

\section{Problem Formulation.}

\subsection{Notation and Preliminaries.}

Let $\mathcal{G}:=(\mathcal{V},\mathcal{E})$ be an undirected unweighted graph
(or network), where $\mathcal{V}:=\{v_{1},v_{2},...v_{n}\}$ represents
the set of nodes and $\mathcal{E}\subseteq\mathcal{V}\times\mathcal{V}$
represents the set of edges%
\footnote{We consider only graphs with no self-loops (i.e., edges of the type $\left\{ v_{i},v_{i}\right\} $).%
}. The \emph{neighborhood} of $v_{i}\in\mathcal{V}$ is defined as
$\mathcal{N}_{i}:=\{v_{j}\in\mathcal{V}:\{v_{i},v_{j}\}\in\mathcal{E}\}$.
The \emph{degree} of node $v_{i}$ is $d_{i}:=|\mathcal{N}_{i}|$.
A \emph{weighted graph} is defined as the triad $\mathcal{H}:=(\mathcal{V},\mathcal{E},\mathcal{W})$,
where $\mathcal{W}$ is a weight function $\mathcal{W}:\mathcal{E}\to\mathbb{R}_{+}$
that assigns a real positive weight to each edge in $\mathcal{E}$.
We define the \emph{weight coefficient} as $w_{ij}:=\mathcal{W}(\{v_{i},v_{j}\})$
if $\{v_{i},v_{j}\}\in\mathcal{E}$, and $w_{ij}=0$ otherwise. The
weighted degree of node $v_{i}$ in a weighted graphs is defined as
$d_{i}^{\mathcal{H}}:=\sum_{j=1}^{n}w_{ij}$.

A \emph{walk} of length $k$ from node $v_{i_{0}}$ to $v_{i_{k}}$
is defined as an ordered sequence of vertices $p:=(v_{i_{0}},v_{i_{1}},\ldots,v_{i_{k}})$,
where $\{v_{i_{l}},v_{i_{l+1}}\}\in\mathcal{E}$, $l=0,1,\ldots,k-1$.
If $v_{i_{0}}=v_{i_{k}}$, the walk is said to be \emph{closed}.
Given a walk $p=(v_{i_{0}},v_{i_{1}},\ldots,v_{i_{k}})$
in a weighted graph $\mathcal{H}$, its weight is defined as the product
of the edge weights, $\omega_{\mathcal{H}}(p):=w_{i_{0}i_{1}}w_{i_{1}i_{2}}....w_{i_{k-1}i_{k}}$.
The \emph{distance} $\delta_{ij}$ between nodes $v_{i}$ and $v_{j}$
is defined as the minimum number of hops from $v_{i}$ to $v_{j}$.

The \emph{adjacency matrix} of an unweighted network $\mathcal{G}$
is defined as the $n\times n$ Boolean symmetric matrix $A_{\mathcal{G}}:=[a_{ij}]$,
defined entry-wise as $a_{ij}=1$ if $v_{i}$ is connected to $v_{j}$,
and $a_{ij}=0$ otherwise.\emph{ }The adjacency matrix of a weighted
graph $\mathcal{H}$ is defined as the symmetric matrix $W_{\mathcal{H}}:=[w_{ij}]$,
where $w_{ij}$ are the weight coefficients. The\emph{ degree matrix}
of a weighted graph $\mathcal{H}$ is the diagonal matrix of its weighted
degrees, i.e., $D_{\mathcal{H}}=\mbox{diag}(d_{i}^{\mathcal{H}})$.
The \emph{normalized Laplacian matrix} of a weighted graph $\mathcal{H}$
is defined as
\begin{equation}
L_{\mathcal{H}}:=I-D_{\mathcal{H}}{}^{-1/2}W_{\mathcal{H}}D_{\mathcal{H}}{}^{-1/2}.\label{eq:Laplacian_Matrix}
\end{equation}
The normalized Laplacian $L_{\mathcal{H}}$ is a symmetric, positive
semidefinite matrix \cite{chung1997spectral}. Thus, it has $n$ nonnegative eigenvalues $\lambda_{1}\geq\lambda_{2}\geq\ldots\geq\lambda_{n}\geq0$
and a full set of orthogonal eigenvectors $v_{1},\ldots,v_{n}$. The
largest eigenvalue $\lambda_{1}$ is called the \emph{spectral radius}
of $L_{\mathcal{H}}$, which satisfies $\lambda_{1}\leq2$, \cite{chung1997spectral}. Given a $n\times n$
symmetric matrix $M$ with (real) eigenvalues $\nu_{1},\ldots,\nu_{n}$,
we define the \emph{$k$-th spectral moment} of $M$ as
\begin{equation}
m_{k}(M):=\frac{1}{n}\sum\limits _{i=1}^{n}\nu_{i}^{k}.
\end{equation}

We now provide graph-theoretical elements to characterize the information
contained in local subgraphs of the network. Given a weighted
graph $\mathcal{H}$, we define the $r$-th order neighborhood around
node $v_{i}$ as the subgraph $\mathcal{H}_{i,r}=(\mathcal{N}_{i,r},\mathcal{E}_{i,r},\mathcal{W})$
with node-set $\mathcal{N}_{i,r}:=\left\{ v_{j}\in\mathcal{V}:\delta_{i,j}\leq r\right\} $
and edge-set $\mathcal{E}_{i,r}:=\{\left\{ v,w\right\} \in\mathcal{E}$
s.t. $v,w\in\mathcal{N}_{i,r}\}$. Notice that $\mathcal{H}_{i,r}$
is the subgraph of $\mathcal{H}$ induced%
\footnote{An induced subgraph is a subset of the vertices of a graph $\mathcal{G}$
together with any edges whose endpoints are both in this subset.%
} by the set of nodes that are at a distance $r$ or less from $v_{i}$.
This set of nodes can be found using a BFS of radius $r$ starting
at node $v_{i}$. Motivated by this interpretation, we call $\mathcal{H}_{i,r}$
the \emph{egonet} of radius $r$ around node $v_{i}$. Egonets can
be algebraically represented via submatrices of the weighted adjacency
matrix $W_{\mathcal{H}}$, as follows. Given a set of $k$ nodes $\mathcal{K}\subseteq\mathcal{V}$,
we denote by $W_{\mathcal{H}}\left(\mathcal{K}\right)$ the $k\times k$
submatrix of $W_{\mathcal{H}}$ formed by selecting the rows and columns
of $W_{\mathcal{H}}$ indexed by $\mathcal{K}$. In particular, given
the egonet $\mathcal{H}_{i,r}$, we define its adjacency matrix as
$W_{i,r}(\mathcal{H}):=W_{\mathcal{H}}\left(\mathcal{N}_{i,r}\right)$. By convention,
we associate the first row and column of the submatrix $W_{i,r}$
with node $v_{i}\in\mathcal{V}$ (the `center' of the egonet), which
can be done via a simple permutation of $W_{i,r}$.

\subsection{Problem Statement.}

The Laplacian eigenvalues of a graph can be efficiently computed for
graphs of small and medium size. In graphs of large size, this computation
is much more challenging. Furthermore, in many real-world networks,
one cannot retrieve the complete network structure due to, for example, privacy and/or security constrains  (e.g. Facebook). Alternatively,
it is usually easy to retrieve local samples of the graph structure in the
form of egonets. For example, one can acquire information about the network structure by extracting egonets of radius $r$ around a random sample of nodes. Therefore, it is realistic to assume that one does not have access to the complete topology of a large-scale network; instead, one can access only a (relatively small) number of egonets in the network.

Clearly, egonets do not completely describe the network structure;
thus, it is impossible to compute exactly the graph spectrum from local egonets.
In this paper, we show that, despite this limitation, we are able
to compute many spectral graph properties from the egonets.
 We show that given a (sufficiently large) random collection
of egonets of radius $r$, one can efficiently estimate the spectral
moments of the normalized Laplacian matrix, $m_{k}(L_{\mathcal{H}})$, for $k\leq2r+1$.
Furthermore, we show that, given a truncated sequence of spectral moments, one can derive bounds on relevant spectral
properties, such as the spectral radius of the Laplacian. As part of our analysis, we provide quality guarantees
for all the estimators and bounds herein proposed.

\section{Spectral Moments for Random Egonets}
We start our analysis assuming the (unrealistic) situation in which one can access all the egonets in
the network. Under this assumption, we shall derive expressions for the spectral moments
of the normalized Laplacian matrix. Afterwards, we shall relax our assumptions and
consider the more realistic case in which one can only access a (relatively small) number of random egonets. In this case, we propose estimators for the spectral moments
and analyze their quality using Hoeffding inequality.

\subsection{Spectral Moments as Averages.}

We derive an expression for the $k$-th spectral moments of the normalized Laplacian
matrix of a weighted graph, $L_{\mathcal{H}}$, from local egonets
of radius $r$, $\mathcal{H}_{i,r}$. In our derivations, we use the
following lemma from algebraic graph theory \cite{biggs1993algebraic}:

\begin{lemma}\label{WalkWeight} Let $\mathcal{H}$ be an undirected,
weighted graph with adjacency matrix $W_{\mathcal{H}}=[w_{ij}]$,
then
\begin{equation}
\left[W_{\mathcal{H}}^{k}\right]_{ii}=\sum_{p\in P_{i,k}}\omega_{\mathcal{H}}(p),
\end{equation}
where $\left[W_{\mathcal{H}}^{k}\right]_{ii}$ is the $\left(i,i\right)$-th
entry of the $k$-th power of the adjacency $W_{\mathcal{H}}$ and
$P_{i,k}$ is the set of all closed walks of length $k$ starting
and finishing at node $v_{i}$.\end{lemma}

Using Lemma \ref{WalkWeight}, we can compute the spectral moments
of the weighted adjacency matrix $W_{\mathcal{H}}$, as follows:

\begin{theorem}\label{Moments from Subgraphs}Consider a weighted,
undirected graph $\mathcal{H}$ with adjacency matrix $W_{\mathcal{H}}$.
Let $W_{i,r}(\mathcal{H})$ be the (weighted) adjacency matrix of the egonet of
radius $r$ around node $v_{i}$, $\mathcal{H}_{i,r}$. Then, the
spectral moments of $W_{\mathcal{H}}$ can be written as
\begin{equation}
m_{k}\left(W_{\mathcal{H}}\right)=\frac{1}{n}\sum_{i=1}^{n}\left[W_{i,r}^{k}\left(\mathcal{H}\right)\right]_{11},\label{Moments as Sum of Traces}
\end{equation}
for $k\leq2r+1$.\end{theorem}

\medskip{}

\begin{proof}Since the trace of a matrix is the sum of its eigenvalues,
we can expand the $k$-th spectral moment of the adjacency matrix
as follows:
\begin{equation}
m_{k}\left(W_{\mathcal{H}}\right)=\frac{1}{n}\text{Trace}\left(W_{\mathcal{H}}^{k}\right)=\frac{1}{n}\sum_{i=1}^{n}\left[W_{\mathcal{H}}^{k}\right]_{ii}.\label{Moments as Walks in Laplacian Graph}
\end{equation}
From Lemma \ref{WalkWeight}, we have that $\left[W_{\mathcal{H}}^{k}\right]_{ii}=\sum_{p\in P_{k,i}}\omega_{\mathcal{H}}(p)$.
Notice that for a fixed value of $k$, closed walks of length $k$
in $\mathcal{H}$ starting at node $v_{i}$ can only touch nodes within
a certain distance $r\left(k\right)$ of $v_{i}$, where $r\left(k\right)$
is a function of $k$. In particular, for $k$ even (resp. odd), a
closed walk of length $k$ starting at node $i$ can only touch nodes
at most $k/2$ (resp. $\left\lfloor k/2\right\rfloor $) hops away
from $v_i$. Therefore, closed walks of length $k$ starting
at $v_i$ are always contained within the neighborhood of radius $\left\lfloor k/2\right\rfloor $.
In other words, the egonet $\mathcal{H}_{i,r}$ of radius $r$ contains
all closed walks of length up to $2r+1$ starting at node $v_i$. We
can count these walks by applying Lemma \ref{WalkWeight}
to the local adjacency matrix $W_{i,r}$. In particular, $\sum_{p\in P_{k,i}}\omega_{\mathcal{H}}(p)$
is equal to $\left[W_{i,r}^{k}\left(\mathcal{H}\right)\right]_{11}$
(since, by convention, node $1$ in the local egonet $\mathcal{H}_{i,r}$
corresponds to node $i$ in the graph $\mathcal{H}$). Therefore,
for $k\leq2r+1$, we have that
\begin{equation}
\left[W_{i,r}^{k}\left(\mathcal{H}\right)\right]_{11}=\sum_{p\in P_{k,i}}\omega_{\mathcal{H}}(p)=\left[W_{\mathcal{H}}^{k}\right]_{ii}.\label{Walks as Diagonal of Sublaplacian}
\end{equation}
Then, substituting (\ref{Walks as Diagonal of Sublaplacian}) into
(\ref{Moments as Walks in Laplacian Graph}), we obtain the statement
of our Theorem.\end{proof}

The above theorem allows us to compute a truncated sequence of spectral
moments $\left\{ m_{k}\left(W_{\mathcal{H}}\right)\text{, }k\leq2r+1\right\} $,
given all the egonets of radius $r$, $\left\{ \mathcal{H}_{i,r}\text{, }v_{i}\in\mathcal{V}\right\} $.
According to (\ref{Moments as Sum of Traces}), the $k$-th spectral
moment is simply the average of the quantities $\left[W_{i,r}^{k}\left(\mathcal{H}\right)\right]_{11}$,
$i=1,...,n$. For a fixed $k$, each value $\left[W_{i,r}^{k}\left(\mathcal{H}\right)\right]_{11}$,
$i=1,\dots,n$, can be computed in time $O\left(\left\vert \mathcal{N}_{i,r}\right\vert ^{3}\right)$,
where $\left\vert \mathcal{N}_{i,r}\right\vert $ is the number of
nodes in the local egonet $\mathcal{H}_{i,r}$. Notice that, if $\left\vert \mathcal{N}_{i,r}\right\vert =o\left(n\right)$,
we can compute the $k$-th spectral moments in linear time (with respect
to the size of the network) using (\ref{Moments as Sum of Traces}).

In what follows, we use the above results to compute the spectral
moments of the Laplacian matrix of a weighted graph, $L_{\mathcal{H}}$.
Before we present our results, we define the so-called Laplacian graph:

\begin{Definition}Given a weighted graph $\mathcal{H}$, we define
its Laplacian graph as $\mathcal{L}\left(\mathcal{H}\right):=\left(\mathcal{V},\mathcal{E\cup\widetilde{E}},\Pi\right)$,
where $\mathcal{\widetilde{E}}:=\left\{ \left\{ v,v\right\} :v\in\mathcal{V}\right\} $
(the set of all self-loops), and the weight function $\Pi:\mathcal{E\cup\widetilde{E}}\to\mathbb{R}$
is defined as:
\begin{equation}
\Pi\left(\left\{ v_{i},v_{j}\right\} \right):=\left\{ \begin{array}{ll}
1, & \text{for }v_{i}=v_{j},\\
\frac{-w_{ij}}{\sqrt{d_{i}^{\mathcal{H}}d_{j}^{\mathcal{H}}}}, & \text{for }\left\{ v_{i},v_{j}\right\} \in\mathcal{E}, v_{i}\neq v_{j}\\
0, & \text{otherwise,}
\end{array}\right.\label{eq:Laplacian_Graph}
\end{equation}
where $d_{i}^{\mathcal{H}}$ is the weighted degree of node $v_{i}$
in $\mathcal{H}$.\end{Definition}

Notice that the weighted adjacency matrix of the Laplacian graph $\mathcal{L}\left(\mathcal{H}\right)$,
denoted by $W_{\mathcal{L}\left(\mathcal{H}\right)}$, is equal to
the normalized Laplacian matrix of the weighted graph $\mathcal{H}$,
$L_{\mathcal{H}}$. Thus, $m_{k}(L_{\mathcal{H}})=m_{k}\left(W_{\mathcal{L}(\mathcal{H})}\right)$
and we can compute the spectral moments of the normalized Laplacian matrix using
weighted walks in the Laplacian graph.
In particular, the Laplacian
spectral moments satisfy
\begin{eqnarray}
m_{k}(L_{\mathcal{H}}) & = & m_{k}\left(W_{\mathcal{L}(\mathcal{H})}\right)\nonumber \\
 & = & \frac{1}{n}\sum_{i=1}^{n}\left[W_{\mathcal{L\left(H\right)}}^{k}\right]_{ii}\nonumber \\
 & = & \frac{1}{n}\sum_{i=1}^{n}\left[W_{i,r}^{k}\left(\mathcal{L\left(H\right)}\right)\right]_{11}\label{eq:Moments as local averages}
\end{eqnarray}
where $\left[W_{i,r}^{k}\left(\mathcal{L\left(H\right)}\right)\right]_{11}$
is the $\left(1,1\right)$-th entry of the $k$-th power of the weighted
adjacency matrix representing the egonet of radius $r$ around node
$v_{i}$ in the Laplacian graph $\mathcal{L\left(H\right)}$. Notice
that $\left[W_{i,r}^{k}\left(\mathcal{L\left(H\right)}\right)\right]_{11}$
is a real number that depends solely on the structure of the egonet
$W_{i,r}(\mathcal{H})$; thus, it is a variable that can be computed using local information about the structure of the network around node $v_{i}$.

In theory, if we had access to all the egonets in the graph, we could
calculate $W_{i,r}^{k}\left(\mathcal{L\left(H\right)}\right)$ for
all $v_{i}\in\mathcal{V}$, and compute the spectral moments of the
Laplacian matrix in linear time (under certain sparsity assumptions).
However, it is often impractical to traverse all the egonets of a
real-world large-scale network because of high computational cost. In the following subsection, we introduce
a method to approximate the spectral moments of a network in sublinear time from a random
collection of $s$ egonets and analyze the quality of our approximation
as a function of $s$.

\subsection{Sampling Egonets and Moment Estimation.\label{sub:Moment Estimators from Samples}}

Define the following `local' variable
\begin{equation}
\phi_{i,r}^{\left(k\right)}  :=  \left[W_{i,r}^{k}\left(\mathcal{L\left(H\right)}\right)\right]_{11}, k\leq 2r+1.\label{eq:Local Phi}
\end{equation}
Notice that $\phi_{i,r}^{\left(k\right)}$ is a function of the egonet
of radius $r$ around node $v_{i}$, since $W_{i,r}\left(\mathcal{L\left(H\right)}\right)$
is the weighted adjacency matrix of the egonet in the Laplacian graph.
Thus, $\phi_{i,r}^{\left(k\right)}$ is a local variable associated
to the $r$-th neighborhood around node $v_{i}$. According to (\ref{eq:Moments as local averages}), the $k$-th Laplacian spectral moment can be computed as the average,
\begin{equation}
m_{k}(L_{\mathcal{H}})  =  \frac{1}{n}\sum_{i=1}^{n}\phi_{i,r}^{\left(k\right)},\nonumber
\end{equation}
Let us now assume that we
do not have access to $\phi_{i,r}^{\left(k\right)}$, for all $v_{i}\in\mathcal{V}$;
instead, we only have access to $\phi_{i,r}^{\left(k\right)}$ for
$v_{i}\in\mathcal{S}$, where $\mathcal{S}\subset\mathcal{V}$ is a subset
of randomly sampled nodes. Since the spectral moment is a global average, we propose the following estimator of $m_{k}(L_{\mathcal{H}})$:
\begin{equation}
\widetilde{m}_{k}(L_{\mathcal{H}}):=\frac{1}{\left|\mathcal{S}\right|}\sum_{v_{i}\in\mathcal{S}}\phi_{i,r}^{\left(k\right)}.\label{eq:Estimator Moments}
\end{equation}
In what follows, we establish the quality of this estimator using
Chernoff-Hoeffding inequality.

\begin{lemma}\label{Lemma:Hoeffding} (Hoeffding Inequality) Let
$X_{1},X_{2},\ldots,X_{k}$ be independent random variables with $P(X_{i}\in[a,b])=1$
for $1\leq i\leq k$. Define the mean of these variables as $\overline{X}=\frac{1}{k}\sum_{i=1}^{k}X_{i}$,
then for any positive $t$, the following inequality holds
\begin{equation}
Pr\{\left|\overline{X}-\mathbb{E}(\overline{X})\right|\geq t\}\leq2\exp\left(\frac{-2kt^{2}}{(b-a)^{2}}\right),
\end{equation}
where $\mathbb{E}(\overline{X})$ is the expected value of $\overline{X}$.\end{lemma}

In order to apply the above lemma, we need the following result:

\begin{lemma}\label{Lemma:Bounded phis}The variable $\phi_{i,r}^{\left(k\right)}$
satisfies $0\leq\phi_{i,r}^{\left(k\right)}\leq2^{k-1}$ for all $i\in\left[n\right]$,
$k\geq1$.\end{lemma}

\begin{proof}Let $\lambda_{i}$ and $\mathbf{v}_{i}$ denote the
eigenvalues and eigenvectors of the Laplacian matrix $L_{\mathcal{H}}$,
for $i=1,2,\ldots,n$. Since $L_{\mathcal{H}}$ is symmetric is always diagonalizable and it has a complete set of orthonormal eigenvectors. Furthermore, $L_{\mathcal{H}}$ is also positive semidefinite; thus, its eigenvalues are nonnegative. Define the matrix
$V$ whose columns are the eigenvectors $\mathbf{v}_{i}$, and the
diagonal matrix $\Lambda=\mbox{diag}(\lambda_{1},\lambda_{2}...,\lambda_{n})$.
Then, $L_{\mathcal{H}}=V\Lambda V^{T}$, and
\begin{equation}
L_{\mathcal{H}}^{k}=V\Lambda^{k}V^{T}=\sum\limits _{j=1}^{n}\lambda_{j}^{k}\mathbf{v}_{j}\mathbf{v}_{j}^{T}.
\end{equation}
Denoting the $i$-th element of vector $\mathbf{v}_{j}$ as $v_{j,i}$,
we have
\begin{equation}
[L_{\mathcal{H}}^{k}]_{ii}=\sum\limits _{j=1}^{n}\lambda_{j}^{k}v_{j,i}^{2}.\label{decomposition-1}
\end{equation}
From (\ref{eq:Laplacian_Matrix}), we have that $[L_{\mathcal{H}}]_{ii}=1$.
Thus,
\begin{equation}
[L_{\mathcal{H}}]_{ii}=\sum\limits _{j=1}^{n}\lambda_{j}v_{j,i}^{2}=1.
\end{equation}
According to \cite{chung1997spectral}, the eigenvalues satisfy $0\leq\lambda_{i}\leq2$
for any $i\in\left[n\right]$. Thus,
\begin{eqnarray*}
[L_{\mathcal{H}}^{k}]_{ii} & = & \sum\limits _{j=1}^{n}\lambda_{j}v_{j,i}^{2}\cdot\lambda_{j}^{k-1}\\
 & \leq & \left(\sum\limits _{j=1}^{n}\lambda_{j}v_{j,i}^{2}\right)\lambda_{1}^{k-1}\\
 & \leq & \left(\sum\limits _{j=1}^{n}\lambda_{j}v_{j,i}^{2}\right)2^{k-1}\\
 & = & 2^{k-1},
\end{eqnarray*}
where we have used the fact that $2\geq\lambda_{1}\geq\lambda_{i}$,
for all $i$. Also, notice from (\ref{decomposition-1}), that every
element in the summation is nonnegative, then $[L_{\mathcal{H}}^{k}]_{ii}$
is nonnegative.\end{proof}

From Lemmas \ref{Lemma:Bounded phis} and \ref{Lemma:Hoeffding},
we obtain the following quality guarantee on our estimator:

\begin{theorem}\label{Thm:Moment Prob Inequality}Consider a set
$\mathcal{S}\subset\mathcal{V}$ of nodes chosen uniformly at random.
Then, the estimator $\widetilde{m}_{k}(L_{\mathcal{H}})$ for the
$k$-th Laplacian spectral moment defined in (\ref{eq:Estimator Moments})
satisfies
\[
\Pr\left\{ \left|\widetilde{m}_{k}(L_{\mathcal{H}})-m_{k}(L_{\mathcal{H}})\right|\geq t_{k}\right\} \leq2\exp\left(\frac{-8t_{k}^{2}\left|\mathcal{S}\right|}{4^{k}}\right).
\]

\end{theorem}

\begin{proof}The proof is a direct application of Lemma \ref{Lemma:Hoeffding}
after substituting $\phi_{i,r}^{\left(k\right)}$ for $X_{i}$ and
$\left[0,2^{k-1}\right]$ for $\left[a,b\right]$. Using this result, we can calculate the number of samples $\left|\mathcal{S}\right|$
needed to achieve a particular error in our moment estimation with
a given probability. For each value of $k$, we denote by $s_{k}$
the sample size needed to achieve an error $t_{k}$ with a probability
less or equal to $\delta_{k}=2\exp\Big(\frac{-8t_{k}^{2}s_{k}}{4^{k}}\Big)$. Let us define the \emph{normalized
error} $\varepsilon_{k}:=\frac{t_{k}}{2^{k-1}}$, then taking $s_{k}=\frac{1}{2}\varepsilon_{k}^{-2}\ln\frac{2}{\delta_{k}}$
samples, we achieve an error $t_{k}$ with probability at most $\delta_{k}$,
i.e., $\Pr\left\{ \left|\widetilde{m}_{k}(L_{\mathcal{H}})-m_{k}(L_{\mathcal{H}})\right|\geq t_{k}\right\} \leq\delta_{k}$. \end{proof}

\section{Moment-Based Spectral Analysis.}

Using Theorem \ref{Moments from Subgraphs}, we can get a truncated
sequence of approximated spectral moments of the Laplacian matrix $L_{\mathcal{H}}$,
$\left\{ m_{k}(L_{\mathcal{H}})\right\} _{k\leq2r+1}$, from a set
of local egonets of radius $r$. We now present a convex optimization
framework to extract information about the largest eigenvalue of the weighted
adjacency matrix, $\lambda_{1}\left(L_{\mathcal{H}}\right)$, from
this sequence of moments.

\subsection{Moment-Based Spectral Bounds}

We can state the problem solved in this subsection as follows:

\textbf{Problem}. \label{Bound from Moments problem}Given a truncated
sequence of Laplacian spectral moments $\left\{ m_{k}(L_{\mathcal{H}})\right\} _{k\leq2r+1}$,
find tight upper and lower bounds on the largest eigenvalue $\lambda_{1}\left(L_{\mathcal{H}}\right)$.

\medskip{}

Our approach is based on a probabilistic interpretation of the eigenvalue
spectrum of a given network. To present our approach, we first need
to introduce some concepts:

\medskip{}

\begin{Definition} \label{Spectral density}Given a weighted, undirected
Laplacian matrix $L_{\mathcal{H}}$ with (real) eigenvalues $\lambda_{1},...,\lambda_{n}$,
the \emph{Laplacian spectral density} is defined as,
\begin{equation}
\mu_{L_{\mathcal{H}}}\left(x\right)\triangleq\frac{1}{n}\sum_{i=1}^{n}\delta\left(x-\lambda_{i}\right),\label{Spectral Measure}
\end{equation}
where $\delta\left(\cdot\right)$ is the Dirac delta function. \end{Definition}

\medskip{}

The spectral density can be interpreted as a discrete probability
density function with \emph{support}%
\footnote{Recall that the support of a finite Borel measure $\mu$ on $\mathbb{R}$,
denoted by $supp\left(\mu\right)$, is the smallest closed set $B$
such that $\mu\left(\mathbb{R}\backslash B\right)=0$.%
} on the set of eigenvalues $\left\{ \lambda_{i},\text{ }i=1...n\right\} $.
Let us consider a discrete random variable $X$ whose probability
density function is $\mu_{L_{\mathcal{H}}}$. The moments of this random
variable satisfy the following \cite{VictorTAC}:
\[
\mathbb{E}_{\mu_{L_{\mathcal{H}}}}\left(X^{k}\right)=m_{k}\left(L_{\mathcal{H}}\right),
\]
for all $k\geq0$.

\medskip{}

We now present a convex optimization framework that allows us to find
bounds on the endpoints of the smallest interval $\left[a,b\right]$
containing the support of a generic random variable $X\sim\mu$ given
a sequence of moments $\left(M_{0},M_{1},...,M_{2r+1}\right)$, where
$M_{k}\triangleq\int x^{k}dx$. Subsequently, we shall apply these
results to find bounds on $\lambda_{1}\left(L_{\mathcal{H}}\right)$.
Our formulation is based on the following matrices:

\medskip{}

\begin{Definition} \label{Hankel Moment Matrices}Given a sequence
of moments $\mathbf{M}_{2r+1}=\left(M_{0},M_{1},...,M_{2r+1}\right)$,
let $H_{2r}\left(\mathbf{M}_{2r+1}\right)$ and $H_{2r+1}\left(\mathbf{M}_{2r+1}\right)\in\mathbb{R}^{\left(r+1\right)\times\left(r+1\right)}$
be the Hankel matrices defined by%
\footnote{For simplicity in the notation, we shall omit the argument $\mathbf{M}_{2r+1}$
whenever clear from the context.%
}:
\begin{eqnarray}
\left[H_{2r}\right]_{ij}\triangleq M_{i+j-2},
&\left[H_{2r+1}\right]_{ij}\triangleq M_{i+j-1}\label{eq:Moment Matrices2}.
\end{eqnarray}
The above matrices are called the \emph{moment matrices} associated
with the sequence $\mathbf{M}_{2r+1}$. \end{Definition}

Given a truncated sequence of moments of a probability distribution,
we can compute a bound on its support as follows \cite{lasserre2011bounding}\cite{lasserre2009moments}:

\medskip{}

\begin{theorem} \label{Main Theorem for general densities}Let $\mu$
be a probability density function on $\mathbb{R}$ with associated
sequence of moments $\mathbf{M}_{2r+1}=\left(M_{0},M_{1},...,M_{2r+1}\right)$,
all finite, and let $\left[a,b\right]$ be the smallest interval which
contains the support of $\mu$. Then, $b\geq\beta^{\ast}\left(\mathbf{M}_{2r+1}\right)$,
where
\begin{equation}
\begin{array}{rrl}
\beta_{r}^{\ast}\left(\mathbf{M}_{2r+1}\right):= & \min_{x} & x\\
 & \text{s.t.} & H_{2r}\succeq0,\\
 &  & x~H_{2r}-H_{2r+1}\succeq0.
\end{array}\label{SDP generic bound}
\end{equation}
\end{theorem}

\medskip{}

Observe that, for a given sequence of moments $\mathbf{M}_{2r+1}$,
the entries of $xH_{2r}-H_{2r+1}$ depend affinely on the variable
$x$. Then $\beta^{\ast}\left(\mathbf{m}_{2r+1}\right)$ is the solutions
to a semidefinite program%
\footnote{A semidefinite program is a convex optimization problem that can be
solved in time polynomial in the input size of the problem; see e.g.
\cite{vandenberghe1996semidefinite}.%
} (SDP) in one variable. Hence, $\beta^{\ast}\left(\mathbf{M}_{2r+1}\right)$
can be efficiently computed using standard optimization software,
e.g. \cite{grant2008cvx}, from a truncated sequence of moments.

Applying Theorem \ref{Main Theorem for general densities} to the
spectral density $\mu_{L_{\mathcal{H}}}$ of a given graph $\mathcal{H}$
with spectral moments $\left(m_{0},m_{1},...,m_{2r+1}\right)$, we
can find a lower bound on its largest eigenvalue, $\lambda_{1}\left(L_{\mathcal{H}}\right)$,
as follows \cite{VictorTAC}:

\medskip{}

\begin{theorem} \label{Main Theorem for Spectral Densities}Let $L_{\mathcal{H}}$
be the normalized Laplacian matrix of a weighted, undirected graph with (real)
eigenvalues $\lambda_{1}\geq...\geq\lambda_{n}$. Then, given a truncated
sequence of the spectral moments of $L_{\mathcal{H}}$, $\mathbf{m}_{2r+1}=\left(m_{0},m_{1},\ldots,m_{2r+1}\right)$,
we have that
\begin{equation}
\lambda_{1}\left(L_{\mathcal{H}}\right)\geq\beta_{r}^{\ast}\left(\mathbf{m}_{2r+1}\right),\label{SDP spectral bound}
\end{equation}
where $\beta_{r}^{\ast}\left(\mathbf{m}_{2r+1}\right)$ is the solution
to the SDP in (\ref{SDP generic bound}). \end{theorem}

\medskip{}

Using the optimization framework presented above, we can also compute
upper bounds on the spectral radius of $\mathcal{H}$ from a sequence
of its spectral moments, as follows. In this case, our formulation
is based on the following set of Hankel matrices:

\medskip{}

\begin{Definition} Given the Laplacian matrix of a weighted, undirected
graph $L_{\mathcal{H}}$ with $n$ nodes and spectral moments $\mathbf{m}_{2r+1}=\left(m_{0},m_{1},...,m_{2r+1}\right)$,
let $T_{2r}\left(y;\mathbf{m}_{2r+1},n\right)$ and $T_{2r+1}\left(y;\mathbf{m}_{2r+1},n\right)\in\mathbb{R}^{\left(r+1\right)\times\left(r+1\right)}$
be the Hankel matrices defined by%
\footnote{We shall omit the arguments from $T_{2r}$ and $T_{2r+1}$ whenever
clear from the context.%
}:
\begin{eqnarray}
\left[T_{2r}\right]_{ij} := & \frac{n}{n-1}m_{i+j-2}-\frac{1}{n-1}y^{i+j-2},\label{eq:Bulk Hankel Matrices1}\\
\left[T_{2r+1}\right]_{ij} := & \frac{n}{n-1}m_{i+j-1}-\frac{1}{n-1}y^{i+j-1}.\label{eq:Bulk Hankel Matrices2}
\end{eqnarray}
\end{Definition}

\medskip{}

Given a sequence of spectral moments, we can compute upper bounds
on the largest eigenvalue $\lambda_{1}\left(L_{\mathcal{H}}\right)$
using the following result \cite{VictorToN}\cite{VictorTAC}:

\medskip{}

\begin{theorem} \label{Upper Bound for Spectral Densities}Let $L_{\mathcal{H}}$
be the normalized Laplacian matrix of a weighted, undirected graph
with (real) eigenvalues $\lambda_{1}\geq...\geq\lambda_{n}$. Then,
given a truncated sequence of its Laplacian spectral moments $\mathbf{m}_{2r+1}=\left(m_{0},m_{1},...,m_{2r+1}\right)$,
we have that
\[
\lambda_{1}\leq\delta_{r}^{\ast}\left(\mathbf{m}_{2r+1},n\right),
\]
where
\begin{equation}
\begin{array}{rrl}
\delta_{r}^{\ast}\left(\mathbf{m}_{2r+1},n\right):= & \max_{y} & y\\
 & \text{s.t.} & T_{2r}\succeq0,\\
 &  & yT_{2r}-T_{2r+1}\succeq0,\\
 &  & T_{2r+1}\succeq0.
\end{array}\label{Upper Bound from Moments}
\end{equation}
\end{theorem}

The optimization program in (\ref{Upper Bound from Moments}) is not
an SDP, since the entries of the matrices $T_{2r}\left(y;\mathbf{m}_{2r+1},n\right)$
and $T_{2r+1}\left(y;\mathbf{m}_{2r+1},n\right)$\ are not affine
functions in $y$. Nevertheless, the program is clearly quasiconvex
\cite{boyd2004convex} and can be efficiently solved using a simple
bisection algorithms.

\medskip{}

In summary, using Theorems \ref{Moments from Subgraphs}, \ref{Main Theorem for Spectral Densities},
and \ref{Upper Bound for Spectral Densities}, we can compute upper
and lower bounds on the largest eigenvalue of the normalized Laplacian matrix
of a weighted, undirected network, $\lambda_{1}\left(L_{\mathcal{H}}\right)$,
from the set of local egonets with radius $r$, as follows: (\emph{Step
1}) Using (\ref{Moments as Sum of Traces}), compute the truncated
sequence of moments $\left\{ m_{k}(L_{\mathcal{H}})\right\} _{k\leq2r+1}$
from the egonets, and (\emph{Step 2}) using Theorems \ref{Main Theorem for Spectral Densities}
and \ref{Upper Bound for Spectral Densities}, compute the upper and
lower bounds, $\delta_{r}^{\ast}\left(\mathbf{m}_{2r+1},n\right)$
and $\beta_{r}^{\ast}\left(\mathbf{m}_{2r+1}\right)$, respectively.
However, the approach presented in this section is based on the assumption that we have
access to all the egonets in the network. In Subsection \ref{sub:Moment Estimators from Samples},
we have provided estimators of the spectral moments from a random sample
of egonets. In the following subsection, we will illustrate how to use
these estimators to derive upper and lower bounds on the Laplacian
spectral radius from a random sample of egonets.

\subsection{Bounds on Spectral Radius from Sampling Egonets}

From Theorem \ref{Thm:Moment Prob Inequality}, we have that the $k$-th
Laplacian spectral moment $m_{k}=m_{k}\left(L_{\mathcal{H}}\right)$
satisfies
\[
\Pr\left\{ m_{k}\in[\widetilde{m}_{k}-t_{k},\widetilde{m}_{k}+t_{k}]\right\} \leq2\exp\left(\frac{-8t_{k}^{2}\left|\mathcal{S}\right|}{4^{k}}\right),
\]
where $\widetilde{m}_{k}=\widetilde{m}_{k}(L_{\mathcal{H}})$ was
defined in (\ref{eq:Estimator Moments})\footnote{We shall omit the arguments $L_{\mathcal{H}}$ from $m_{k}$ and $\widetilde{m}_{k}$ unless there is need for specification.%
}. Then, the probability of
a truncated sequence of moments $\left(m_{k}\right)_{k\leq2r+1}$
satisfying $m_{k}\in[\widetilde{m}_{k}-t_{k},\widetilde{m}_{k}+t_{k}]$
for $k=2,\ldots,2r+1$, satisfies the following proposition (notice
that $m_{1}=1$, for any $L_{\mathcal{H}}$):

\begin{proposition}\label{Prop:Moments in Interval}
For a given $\Delta\in\left[0,1\right]$,
we have that
\[
\Pr\left(\bigcap_{k=2}^{2r+1}\left\{ m_{k}\in[\widetilde{m}_{k}-t_{k},\widetilde{m}_{k}+t_{k}]\right\} \right)\geq\Delta,
\]
if
\begin{equation}
t_{k}=\frac{2^{k-1}}{\sqrt{2\left|\mathcal{S}\right|}}\ln^{1/2}\frac{4r}{1-\Delta}.\label{eq:Definition tk}
\end{equation}
\end{proposition}

\begin{proof}First, we have that

\begin{eqnarray*}
 &  & \Pr\left(\bigcap_{k=2}^{2r+1}\left\{ m_{k}\in[\widetilde{m}_{k}-t_{k},\widetilde{m}_{k}+t_{k}]\right\} \right)\\
 & = & 1-\Pr\left(\bigcup_{k=2}^{2r+1}\left\{ m_{k}\notin[\widetilde{m}_{k}-t_{k},\widetilde{m}_{k}+t_{k}]\right\} \right)\\
 & \geq & 1-\sum_{k=2}^{2r+1}\Pr\left(m_{k}\notin[\widetilde{m}_{k}-t_{k},\widetilde{m}_{k}+t_{k}]\right)\\
 & \geq & 1-2\sum_{k=2}^{2r+1}\exp\left(\frac{-8t_{k}^{2}\left|\mathcal{S}\right|}{4^{k}}\right).
\end{eqnarray*}
The last probability can be made equal to a desired $\Delta$ by choosing
$t_{k}^{2}/4^{k}=\alpha$ for all $k=2,\ldots,2k+1$ with $\alpha$
satisfying $\Delta=1-4r\exp\left(-8\alpha\left|\mathcal{S}\right|\right)$.
Or equivalently, $\alpha=\frac{1}{8\left|\mathcal{S}\right|}\ln\frac{4r}{1-\Delta}$,
which implies the statement of the Proposition, after simple algebraic
manipulations.

\end{proof}

We can then apply the result in Theorem \ref{Main Theorem for Spectral Densities}
to compute a probabilistic lower bound on the spectral radius by solving
a modified version of the SDP in (\ref{SDP generic bound}), as follows.
First, given a sample set of nodes $\mathcal{S}$, we extract the
corresponding egonets of radius $r$. Then, using (\ref{eq:Estimator Moments})
and (\ref{eq:Local Phi}), we compute a sequence of estimators $\widetilde{m}_{k}$
for $k=2,\ldots,2r+1$. Finally, according to Proposition \ref{Prop:Moments in Interval},
we modify the SDP in (\ref{SDP generic bound}) to obtain our main
result.

\begin{theorem} Given a uniform sample set $\mathcal{S}\subset\mathcal{V}$
and the egonets of radius $r$ around the nodes in $\mathcal{S}$,
the spectral radius of the normalized Laplacian matrix satisfies
\[
\Pr\left(\lambda_{1}\left(L_{\mathcal{H}}\right)\geq\widetilde{\beta}_{r}\right)\geq\Delta,
\]
where
\[
\begin{array}{rrl}
\widetilde{\beta}_{r}:= & \min_{x} & x\\
 & \text{s.t.} & H_{2r}\succeq0,\\
 &  & x~H_{2r}-H_{2r+1}\succeq0,\\
 &  & \widetilde{m}_{k}-t_{k}\leq m_{k}\leq\widetilde{m}_{k}+t_{k},\\
 &  & k=2,\ldots,2r+1,
\end{array}
\]
with $m_{k}$ and $\widetilde{m}_{k}$ defined in (\ref{eq:Local Phi}) and (\ref{eq:Estimator Moments}), $H_{2r}$ and $H_{2r+1}$ defined in (\ref{eq:Moment Matrices2}),
and $t_{k}$ defined in (\ref{eq:Definition tk}). \end{theorem}

The proof of the above theorem is a direct adaptation of Theorem \ref{Main Theorem for Spectral Densities}
using Proposition \ref{Prop:Moments in Interval}. The same adaptation
can be applied to derive an upper bound on the spectral radius of
the normalized Laplacian matrix from Theorem \ref{Upper Bound for Spectral Densities}
and Proposition \ref{Prop:Moments in Interval}.

\section{Numerical Analysis}

In this section, we will present the numerical analysis of spectral
radius estimation, and verify the quality of estimation based on
sampled nodes. In our simulations, we will use data from the Euro-Email
network \cite{networkdata}. The network is composed of 36,692 nodes, which
are connected by 183,831 edges. Here we consider the network as unweighted,
undirected simple graph. Two nodes are connected as long as either
user sent email to another. From the simple graph, we can construct
the weighted Laplacian graph that corresponds to the simple graph.

To be able to compute the spectral radius, we extract a small network
with 5000 nodes via BFS, so that we could compare the performance
of sampling with the accurate computation. The subgraph with 5000
nodes will be the object of our analysis.


The nodes are assigned with indexes without consideration about their topology.
To get a uniform sampling, the indexes are picked randomly, which
compose the collection of sampled nodes.

In Figure \ref{SampleConfidenceInOne}, we take different number of
samples from the network to estimate the bounds of the spectral radius
based on egonet with radius $r_{0}=3$. In the simulation, the normalized
error bound for the elements in the moment sequence is fixed, i.e. $\frac{t_{k}}{2^{k-1}}=0.08$.
Here we assume that the moments for $k=1,2...5$ can be accurately
computed, because it does not cost much to compute the power of
the Laplacian matrix up to 5-th order. However, for the 6-th and 7-th moment, we take
uniform samples from the whole network and approximate the moment using the estimator proposed in the previous sections, i.e. using average of the sampled egonets to approximate the global average. Thus for each $k>5$,
$m_{k}\in[\widetilde{m_{k}}-t_{k},\widetilde{m_{k}}+t_{k}]$. With the size of the sampled nodes increasing, the quality (accuracy guarantee) of the estimator increases.
\begin{figure}[!hbt]
\centering \includegraphics[width=3.4in]{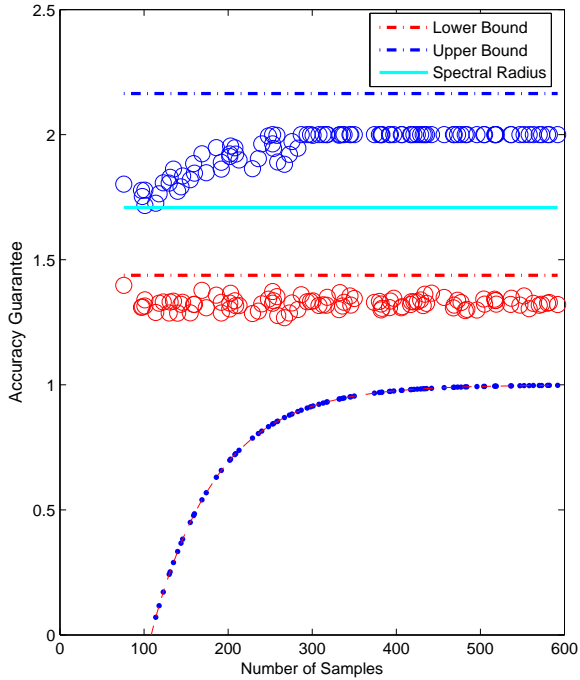}
\caption{Estimation when the size of the samples varies. (1) The normalized error for each moment are the same. Samples with different sizes give different estimations of the bounds. (2) Quality guarantee is a function of the size of the sample. When the number of sample increases, the quality of the estimator increases.}
\label{SampleConfidenceInOne}
\end{figure}


From the upper part of the figure, it can be seen that the lower bound does not change
much when the number of samples changes. For the upper bound, when
the number of samples increases, the bound gets looser, but the accuracy
guarantee that the spectral radius is within the bounds increase.
The dotted lines are the bounds calculated by considering the egonets
of every nodes in the network. And the circles are the estimated bounds when different sets of nodes are taken as samples. The lower part of Figure \ref{SampleConfidenceInOne}
gives the curve for the number of samples versus the accuracy guarantee. Though the network has 5000 nodes, taking 600 samples will
give the estimation with nearly 100 percent.

\begin{figure}[!hbt]
\centering \includegraphics[width=3.6in]{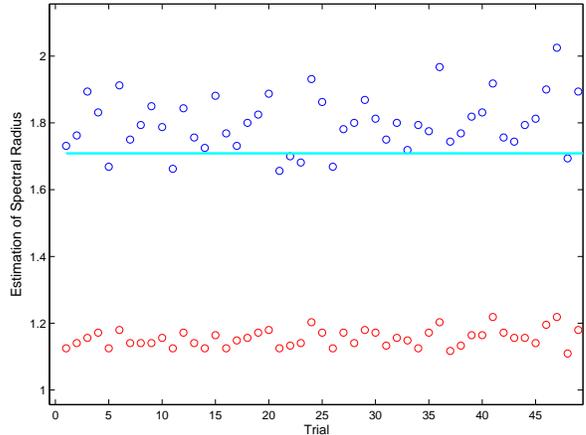}
\caption{Spectral bounds for trials with the same sample size $\mathcal{S}=21$. When the number
of samples are the same, }
\label{Estimation_Confidence4}
\end{figure}

\begin{figure}[!hbt]
\centering \includegraphics[width=3.6in]{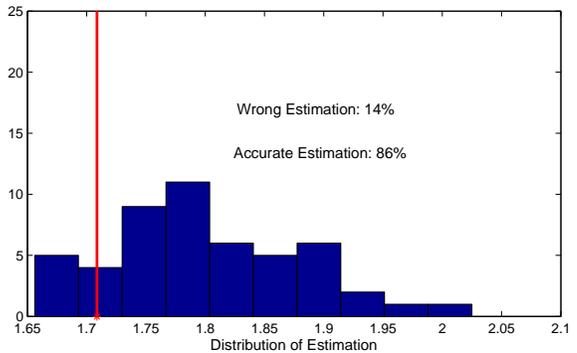}
\caption{Distribution of estimation when the size of the sample is the same.}
\label{Estimation_Confidence44}
\end{figure}

In Figure \ref{Estimation_Confidence4}, we take different samples
with the same sample size to verify the quality of estimation. The
normalized error is set to be $\varepsilon_{k}=0.2$, and $\delta_{k}=0.4$,
thus the sample size needed is $S=21$. From the figure, it can be seen
that the lower bound is much loose, and almost the same when the
sample pool are different. Checking whether the estimation range is
correct for each trial, we can see from Figure \ref{Estimation_Confidence44}
that the accuracy rate is $86\%$, which is much higher than the theoretical
accuracy probability $1-\delta_{k}=0.6$.


\section{Conclusion}

In this paper, we apply graph theories and convex optimization techniques
to study the spectrum property of the normalized Laplacian
matrix. Instead of analyzing the whole network, we focus on localized structural features with radius $r$.

Due to the high cost of traversing all the nodes, we have proposed to take uniform samples from
the network pool and use the sampled egonets to estimate the moments of the normalized Laplacian.
With Hoeffding inequalities, we characterize the quality of the estimators in terms of normalized error and size of the sample. In addition, we have derived the lower and upper bounds of the spectral radius by solving a series of SDP problems, based on the collection of random subgraphs. The combination of quality guarantee of moment sequence and the optimization problems provides us with the estimation guarantee of the spectral radius.

\bibliographystyle{siam}
\bibliography{SDMreference}

\begin{thebibliography}{10}

\bibitem{biggs1993algebraic}
{\sc N.~Biggs}, {\em Algebraic Graph Theory}, Cambridge University Press, 1993.

\bibitem{boyd2004convex}
{\sc S.~Boyd and L.~Vandenberghe}, {\em Convex Optimization}, Cambridge
  university press, 2004.

\bibitem{chung1997spectral}
{\sc F.~Chung}, {\em Spectral Graph Theory}, vol.~92, AMS Bookstore, 1997.

\bibitem{chung2003spectra}
{\sc F.~Chung, L.~Lu, and V.~Vu}, {\em Spectra of random graphs with given
  expected degrees}, Proceedings of the National Academy of Sciences, 100
  (2003), pp.~6313--6318.

\bibitem{cvetkovic2010introduction}
{\sc D.~Cvetkovi{\'c}, P.~Rowlinson, and S.~Simi{\'c}}, {\em An Introduction to
  the Theory of Graph Spectra}, Cambridge University Press Cambridge, 2010.

\bibitem{grant2008cvx}
{\sc M.~Grant, S.~Boyd, and Y.~Ye}, {\em Cvx: Matlab software for disciplined
  convex programming}, 2008.

\bibitem{lasserre2009moments}
{\sc J.-B. Lasserre}, {\em Moments, Positive Polynomials and Their
  Applications}, vol.~1, World Scientific, 2009.

\bibitem{lasserre2011bounding}
\leavevmode\vrule height 2pt depth -1.6pt width 23pt, {\em Bounding the support
  of a measure from its marginal moments}, Proceedings of the American
  Mathematical Society, 139 (2011), pp.~3375--3382.

\bibitem{networkdata}
{\sc J.~Leskovec}, {\em Stanford large network dataset collection}.
\newblock http://snap.stanford.edu/data/index.html.

\bibitem{lovasz1993random}
{\sc L.~Lov{\'a}sz}, {\em Random walks on graphs: A survey}, Combinatorics,
  Paul erdos is eighty, 2 (1993), pp.~1--46.

\bibitem{lynch1996distributed}
{\sc N.A. Lynch}, {\em Distributed algorithms}, Morgan Kaufmann, 1996.

\bibitem{mohar1991laplacian}
{\sc B.~Mohar and Y.~Alavi}, {\em The laplacian spectrum of graphs}, Graph
  theory, combinatorics, and applications, 2 (1991), pp.~871--898.

\bibitem{newman2002assortative}
{\sc M.~Newman}, {\em Assortative mixing in networks}, Physical review letters,
  89 (2002), p.~208701.

\bibitem{newman2009random}
\leavevmode\vrule height 2pt depth -1.6pt width 23pt, {\em Random graphs with
  clustering}, Physical review letters, 103 (2009), p.~058701.

\bibitem{newman2001random}
{\sc M.~Newman, S.~Strogatz, and D.~Watts}, {\em Random graphs with arbitrary
  degree distributions and their applications}, Physical Review E, 64 (2001),
  p.~026118.

\bibitem{olfati2006flocking}
{\sc R.~Olfati-Saber}, {\em Flocking for multi-agent dynamic systems:
  Algorithms and theory}, IEEE Transactions on Automatic Control, 51 (2006),
  pp.~401--420.

\bibitem{pastor2001dynamical}
{\sc R.~Pastor-Satorras, A.~V{\'a}zquez, and A.~Vespignani}, {\em Dynamical and
  correlation properties of the internet}, Physical review letters, 87 (2001),
  p.~258701.

\bibitem{VictorToN}
{\sc V.M Preciado and A.~Jadbabaie}, {\em Moment-based spectral analysis of
  large-scale networks using local structural information}, ACM/IEEE
  Transactions on Networking, 21 (2013), pp.~373--382.

\bibitem{VictorTAC}
{\sc V.M Preciado, A.~Jadbabaie, and G.C. Verghese}, {\em Structural analysis
  of laplacian spectral properties of large-scale networks}, IEEE Transactions
  on Automatic Control, 58 (2013), pp.~2338--2343.

\bibitem{vandenberghe1996semidefinite}
{\sc L.~Vandenberghe and S.~Boyd}, {\em Semidefinite programming}, SIAM review,
  38 (1996), pp.~49--95.

\end{thebibliography}

\end{document}